%% file: pltl.tex
\documentclass{llncs}

\usepackage{amssymb}
\usepackage{amsmath}
\usepackage{relsize}

\usepackage{tikz}
\usepackage{pgf}
\usepackage{tikz}
\usetikzlibrary{arrows,automata}

\input{tlmacros.tex}
\usepackage{algorithm}
\usepackage{algorithmicx}
\usepackage{algpseudocode}
\usepackage{multicol}
\renewcommand{\leq}{\leqslant}
\renewcommand{\geq}{\geqslant}
\renewcommand{\emptyset}{\varnothing}
\renewcommand{\a}[1]{\textbf{\textit{#1}}}

\newcommand{\la}{\langle}
\newcommand{\ra}{\rangle} 
 \newcommand{\ve}{{{\a v}}}
\newcommand{\de}{\Diamond}
\newcommand{\bo}{\Box}

\title{Parametric LTL on Markov Chains}
\author{Souymodip Chakraborty and Joost-Pieter Katoen\thanks{Currently on sabbatical leave at the University of Oxford, United Kingdom.}}
\institute{RWTH Aachen University, Ahornstra\ss{}e 55, D-52074 Aachen, Germany}
 \authorrunning{S. Chakraborty and J.-P. Katoen}

\addtocmark{Hamiltonian Mechanics} 

\begin{document}
 \maketitle

\begin{abstract}
This paper is concerned with the verification of finite Markov chains against parametrized LTL (pLTL) formulas.
In pLTL, the until-modality is equipped with a bound that contains variables; e.g., $\de_{\leq x}\ \varphi$ asserts that $\varphi$ holds within $x$ 
time steps, where $x$ is a variable on natural numbers.
The central problem studied in this paper is to determine the set of parameter valuations $V_{\prec p} (\varphi)$ for which the probability to satisfy pLTL-formula 
$\varphi$ in a Markov chain meets 
a given threshold $\prec p$, where $\prec$ is a comparison on reals and $p$ a probability.
As for pLTL determining the emptiness of $V_{> 0}(\varphi)$ is undecidable, we consider several logic fragments.
We consider parametric reachability properties, 
a sub-logic of pLTL restricted to next and $\de_{\leq x}$, 
parametric B\"uchi properties and finally, a maximal subclass of pLTL for which emptiness of $V_{> 0}(\varphi)$ is decidable.
\end{abstract}

\section{Introduction} 
Verifying a finite Markov chain (MC, for short) $M$ against an LTL-formula $\varphi$ amounts to determining the probability that $M$ satisfies $\varphi$, i.e., 
the likelihood of the set of  infinite paths of $M$ satisfying $\varphi$.
Vardi~\cite{DBLP:conf/focs/Vardi85} considered the qualitative version of this problem, that is, does $M$ almost surely satisfy $\varphi$, or with positive probability.
Together with Wolper, he showed that the qualitative LTL model-checking problem for MCs is PSPACE-complete.
The quantitative verification problem -- what is the probability of satisfying $\varphi$? -- has been treated by Courcoubetis and 
Yannakakis~\cite{DBLP:journals/jacm/CourcoubetisY95}.
An alternative algorithm that has a time complexity which is polynomial in the size of the MC and exponential in $| \varphi |$ is by 
Couvreur \emph{et al.}~\cite{CouvreurSahebSutre03}.
Recently, practical improvements have been obtained by Chatterjee \emph{et al.} for verifying the LTL(F,G)-fragment on MCs using generalized deterministic Rabin automata~\cite{DBLP:conf/cav/ChatterjeeGK13}.

This paper considers the verification of MCs against \emph{parametric} LTL formulas.
In parametric LTL~\cite{DBLP:journals/tocl/AlurETP01} (pLTL, for short), temporal operators can be subscripted by a variable ranging over the natural numbers.
The formula $\de_{\leq x} \, a$ means that in at most $x$ steps $a$ occurs, and $\Box \de_{\leq y} \, a$ means that at every
index $a$ occurs within $y$ steps. 
Note that $x$ and $y$ are variables whose value is not fixed in advance.
The central question is now to determine the values of $x$ and $y$ such that the probability of a given MC satisfying the pLTL-formula $\varphi$ meets a certain 
threshold $p$.
This is referred to as the valuation set $V_{\prec p}(\varphi)$ for comparison operator $\prec$.
This problem has both a qualitative (threshold $> 0$ and $= 1$) and a quantitative variant ($0 < p < 1$).

The main results of this paper are as follows.
Just as for the setting with Kripke structures~\cite{DBLP:journals/tocl/AlurETP01}, it is shown that checking the emptiness of $V_{> 0}(\varphi)$ in 
general is undecidable.
We therefore resort to fragments of pLTL.
We show that determining $V_{\geq p}(\de_{\leq x} \, a)$ can be done by searching in a range defined by the \emph{precision} of the input, 
whereas polynomial time 
graph algorithms suffice for its qualitative variant.
The same applies to formulas of the form $\Box \de_{\leq x} \, a$.
We provide necessary and sufficient criteria for checking the emptiness of $V_{> 0}(\varphi)$ (and $V_{=1}(\varphi)$)
for the fragments pLTL(F,X) and pLTL$_\de$, and prove that checking these criteria are NP-complete and PSPACE-complete, respectively. 
We also define a representation of these sets and provide algorithms to construct them. 


\paragraph{Related work.}
The verification of parametric probabilistic models in which certain transition probabilities are given as parameters (or functions thereof) has recently received 
considerable attention.
Most of these works are focused on parameter synthesis: for which parameter instances does a given (LTL or PCTL) formula hold?
To mention a few, Han \emph{et al.}~\cite{DBLP:conf/rtss/HanKM08} considered this problem for timed reachability in continuous-time MCs, Hahn \emph{et al.}~\cite{DBLP:conf/nfm/HahnHZ11} and
Pugelli \emph{et al.}~\cite{DBLP:conf/cav/Pugelli13} for Markov decision processes (MDPs), and Benedikt \emph{et al.}~\cite{DBLP:conf/tacas/BenediktLW13} for $\omega$-regular properties of interval MCs.
Hahn~\emph{et al.}~\cite{DBLP:journals/sttt/HahnHZ11} provide an algorithm for computing the rational function expressing the probability of reaching a given set of states in a parametric (reward) MDP based on exploiting regular expressions as initially proposed by Daws~\cite{DBLP:conf/ictac/Daws04}.
Other related work includes the synthesis of loop invariants for parametric probabilistic programs~\cite{DBLP:conf/sas/KatoenMMM10}.
To the best of our knowledge, verifying parametric properties on MCs has not been considered so far.
The closest related works are on combining two-variable FO with LTL for MDPs by Benedikt \emph{et al.}~\cite{DBLP:journals/corr/abs-1303-4533} and 
the computation of quantiles by Ummels and Baier~\cite{DBLP:conf/fossacs/UmmelsB13}.  

\paragraph{Organization of the paper.}
Section~\ref{sec:prelim} presents pLTL and MCs and a first undecidability result.
Section~\ref{sec:reach} considers parametric reachability.
Section~\ref{sec:pltlfx} treats the fragment pLTL(F,X) and
Section~\ref{sec:parambuechi} parametric B\"uchi properties.  
Section~\ref{sec:nobox}  treats the bounded always-free fragment of pLTL. 
Section~\ref{sec:concl} concludes the paper. \\[-3ex]

\section{Preliminaries}
\label{sec:prelim}
\vspace{-0.3cm}
\paragraph{Parametric LTL.}
Parametric LTL extends propositional LTL with bounded temporal modalities, for which the bound is either a constant or a variable.
Let $\Var$ be a finite set of variables ranged over by $x, y$, and $\AP$ be a finite set of propositions ranged over by $a$ and $b$.
Let $c \in \Nats$.
Parametric LTL formulas adhere to the following syntax:
$$
\varphi \ ::= \ a \where \neg \varphi \where \varphi \, \wedge \, \varphi  \where \X\! \varphi \where \varphi \U\varphi \where 
\de_{\hspace{-0.05cm}\prec x} \hspace{0.1cm}\varphi \where \de_{\hspace{-0.05cm}\prec c}\hspace{0.1cm} \varphi
$$
where $\prec \, \in \set{=, \leq, <, >, \geq}$.
A pLTL structure is a triple $(w,i,{\a v})$ where $w \in \Sigma^\omega$ with $\Sigma = 2^{\smAP}$ is an infinite word over sets of propositions, 
$i \in \Nats$ is an index, and ${\a v}: \Var \to \Nats$ is a variable valuation.
Analogously, we consider a valuation ${\a v}$ as a vector in $\mathbb{N}^d$, where $d$ for pLTL formula $\varphi$ is the number of variables occurring in $\varphi$. 
E.g. for $d=1$, the valuation is just a number $v$.
We compare valuations $\ve$ and ${\a v}'$ as ${\a v} \leq {\a v}'$ iff ${\a v}(x) \leq {\a v}'(x)$ for all $x$.
Let $w[i]$ denote the $i$-th element of $w$.
The satisfaction relation $\models$ is defined by structural induction over $\varphi$  as follows:
$$
\begin{array}{lcl}
(w,i,{\a v}) \models a & \mbox{iff} & a \in w[i] \\
(w,i,{\a v}) \models \!\Not \varphi & \mbox{iff} & (w,i,{\a v}) \not\models \varphi \\
(w,i,{\a v}) \models \varphi_1 \, \wedge \, \varphi_2 & \mbox{iff} & (w,i,{\a v}) \models\varphi_1 \mbox{ and } (w,i,{\a v}) \models \varphi_2 \\
 (w,i,{\a v}) \models \de_{\hspace{-0.05cm}\prec x} \, \varphi & \mbox{iff} &  
   (w,j,{\a v}) \models \varphi \mbox{ for some } j \prec \ve(x){+}i.  \\
\end{array}
$$  
For the sake of brevity, we have omitted the semantics of the standard LTL modalities. 
As usual, $\varphi_1\Release \varphi_2 \equiv \neg(\neg\varphi_1\U\neg\varphi_2)$, $\de \varphi \equiv \true \U \varphi$ 
and $\bo \varphi \equiv \neg \de \neg \varphi$.
The language of $\varphi$ is defined by $\L(\varphi) = \{(w,{\a v}) \ |\ (w,0,{\a v})\models \varphi\}$. 
Alur et al.~\cite{DBLP:journals/tocl/AlurETP01} have shown that other modalities such as $\U\!_{\leq x}$, $\de_{> x}$, $\Box_{> x}$, $\U\!_{> x}$, 
$\R_{\leq x}$ and $\R_{> x}$, can all be encoded in our syntax.  
For instance, the following equivalences hold: 															
\begin{equation}\label{equiv}
\begin{array}{rclrcl}
\de_{> x} \, \varphi & \, \equiv \, & \Box_{\leq x} \, \de \X \varphi, & 
\Box_{> x} \, \varphi & \, \equiv \, & \de_{\leq x} \, \Box \X \varphi, 
\\
\varphi \U\!_{\leq x} \, \psi & \, \equiv \, & (\varphi \U \psi) \, \wedge \, \de_{\leq x} \, \psi, \quad 
& \varphi \U\!_{> x} \, \psi & \, \equiv \, & \Box_{\leq x} \left( \varphi \, \wedge \, \X (\varphi \U \psi) \right)
\end{array}
\end{equation}
In the remainder of this paper,  we focus on bounded always and eventualities where all bounds are upper bounds.
We abbreviate $\de_{\leq x}$ by $\de_x$ and do similar for the other modalities.
For valuation ${\a v}$ and pLTL-formula $\varphi$, let ${\a v}(\varphi)$ denote the LTL formula obtained 
from $\varphi$ by replacing variable $x$ by its valuation $\ve(x)$; e.g., 
${\a v}(\de_{ x} \, \varphi)$ equals $\de_{ {\a v}(x)} \, {\a v}(\varphi)$.

\paragraph{Markov chains.}
A discrete-time Markov chain $M$ is a quadruple $(S,\bfP,s_0,L)$ where $S$ is a finite set of states with $m=|S|$,
$\bfP: S \times S \to [0,1]$ is a stochastic matrix, $s_0 \in S$ an initial state, and $L: S \to 2^{\smAP}$ a state-labeling function.  
$\bfP(u,v)$ denotes the one-step probability of moving from state $u$ to $v$.
A trajectory (or path) of a Markov chain (MC, for short) $M$ is a sequence $\set{s_i}_{i \geq 0}$ such that $\bfP(s_i,s_{i+1}) > 0$ for all $i \geq 0$.
A trajectory $\pi = s_0 s_1 s_2 \ldots$ induces the trace $\trace(\pi) = L(s_0) L(s_1) L(s_2) \ldots$.
Let $\Paths(M)$ denote the set of paths  of MC $M$.
A path $\pi$ satisfies the pLTL-formula $\varphi$ under the valuation ${\a v}$, denoted $\pi \models{\a v}(\varphi)$,  whenever $(\trace(\pi), 0,{\a v}) \models \varphi$ (or equivalently, $(\trace(\pi),{\a v})\in\mathcal{L}(\varphi)$).
A finite path (or path fragment) satisfies a formula under a valuation if any infinite extension of it also satisfies the formula.
Let $\Pr$ be the probability measure on sets of paths, defined by a standard cylinder construction~\cite{DBLP:conf/focs/Vardi85}.
The probability of satisfying $\varphi$ by $M$ under valuation ${\a v}$ is given by $\Pr \set{\pi \in \Paths(M) \mid \pi \models{\a v}(\varphi)}$, generally abbreviated as $\Pr (M \models \ve(\varphi))$. 

\paragraph{Valuation set.}
The central problem addressed in this paper is to determine the valuation set of a pLTL formula $\varphi$.
Let $M$ be an MC, $p \in [0,1]$ a probability bound, and $\prec \, \in \set{=, \leq , <, >, \geq}$. 
Then we are interested in determining:
$$
V_{\prec p}(\varphi) \ = \ \set{\ve \mid \Pr (M \models{\a v}(\varphi)) \prec p},
$$
i.e., the set of valuations under which the probability of satisfying $\varphi$ meets the bound $\prec p$.
In particular, we will focus on the decidability and complexity of the emptiness problem for $V_{\prec p}(\varphi)$, i.e., 
the decision problem whether $V_{\prec p}(\varphi) = \emptyset$ or not,  on algorithms (if any) determining the set $V_{\prec p}(\varphi)$, 
and on the size of the minimal representation of $V_{\prec p}(\varphi)$.
In the qualitative setting, the bound $\prec p$ is either $> 0$, or $= 1$.

\begin{proposition}\label{undy}
For $\varphi \in$ pLTL, the problem if $V_{> 0}(\varphi) = \emptyset$ is undecidable.
\end{proposition}
\begin{proof}
The proof is based on  \cite[Th.\ 4.1]{DBLP:journals/tocl/AlurETP01}, see the appendix.
\hfill $\quad\blacksquare$
\end{proof}
It follows that deciding whether $V_{=1}(\varphi)=\emptyset$ is undecidable, as $V_{>0}(\varphi)= \emptyset$ iff $V_{=1}(\neg\varphi)\neq \emptyset$.
\noindent
As a combination of $\de_{\leq x}$ and $\Box_{\leq x}$ modalities can encode $\U\!_{=x}$, e.g., 
$$
\neg a \wedge \X (\neg a \U\!_{=x} \, a) 
\ \equiv \ 
\X (\neg a \U\!_{\leq x} \, a) \wedge (\neg a \U\!_{> x} \, a),
$$
we will restrict ourselves to fragments of pLTL where each formula is in negative normal form and 
the only parametrized operator is $\de_{\leq x} \, \varphi$. 
We refer to this fragment as pLTL$_\de$:
\begin{equation}\label{pltl_d}
\varphi \ ::= \ a \where \neg a \where \varphi \wedge  \varphi \where \varphi  \vee  \varphi \where \X\! \varphi \where \varphi \U \varphi 
\where \varphi \Release \varphi 
\where \Box \varphi \where \de_{\leq x}  \, \varphi \where \de_{\leq c} \, \varphi \where \Box_{\leq c} \, \varphi. 
\end{equation}
We show it is a sub-logic of pLTL for which the emptiness problem for $V_{> 0}(\varphi)$ is decidable.  The logic has a favourable \emph{monotonicity} property, i.e.,
\begin{remark}\label{monotonicity}
For every pLTL$_\de$-formula $\varphi$, infinite word $w$ and valuations ${\a v}, {\a v}'$, ${\a v} \leq {\a v}'$ 
implies $(w,{\a v}) \models \varphi \implies (w,{\a v}')\models \varphi$.
\end{remark}
Here $(w,{\a v}) \models \varphi$ is s shorthand for $(w,0,{\a v}) \models \varphi$.
We start off with briefly considering (only) parametric eventualities and then consider the sub-logic pLTL(F,X) restricted to next and $\de_{x}$.
Later on, we also consider parametric B\"uchi formulas, and finally, pLTL$_\de$. Most of the proofs are moved to the appendix.

\section{Parametric Reachability}
\label{sec:reach}

In this section, we consider pLTL-formulas of the form $\de_{x} \, a$ for proposition $a$, or equivalently, $\de_{x} \, T$ for the set of target states $T = \set{s \in S \mid a \in L(s)}$.
We consider bounds of the form $\geq p$ with $0 < p < 1$.
The valuation set of interest is thus $V_{\geq p}(\de_{x} \, a)$.
Let $\mu_i$ be the probability of reaching $T$ within $i$ steps; the sequence $\{\mu_i\}$ is ascending. 
There can be two cases: (a) the sequence reaches a constant value in $m$ steps ($m$ being the size of Markov chain) or (b) the sequence monotonically increases and converges to $\mu_\infty$. 
This makes the emptiness problem for $V_{\geq p}(\de_{x} \, a)$ decidable. 
In the first case, we check $\mu_m\geq p$.
In the second case, emptiness is decidable in time polynomial  in $m$, by determining $\mu_\infty = \Pr (\de a)$ which can be done by solving a system of linear equations with at most $m$ variables.
Then, $V_{\geq p}(\de_x \, a) \neq \emptyset$ iff $p < \mu_\infty$.

Assume in the sequel that $T$ is non-empty.
Let $\min V_{\geq p}(\de_{x} \, a)=n_0$.
The valuation set can thus be represented by $n_0$ (this gives a minimal representation of the set).
Membership queries, i.e., does $n \in V_{\geq p}(\de_{x} \, a)$, then simply boil down to checking whether $n_0 \leq n$, 
which can be done in constant time (modulo the size of $n_0$).
The only catch is that $n_0$ can be very large if $p$ is  close to $\mu_\infty$. 
A simple example elucidates this fact.
\begin{example}
Consider the MC $M$ with $S = \set{s_0, t}$, $L(t) = \set{a}$, $L(s_0) = \emptyset$, $\bfP(s_0,s_0) = \frac{1}{2} = \bfP(s_0,t)$ and $\bfP(t,t) = 1$.
Then $\Pr(M \models \de_n \, a) = 1 - \left( \frac{1}{2} \right)^n$.
It follows that $\min V_{\geq p}(\de_x \, a)$ goes to infinity when $p$ approaches one.
\end{example}
The following bound on $n_0$ can nonetheless be provided.  
This bound allows for obtaining the minimum value $n_0$ by a binary search. 
\begin{proposition} For MC $M$,
 $\min V_{\geq p}(\Diamond_x a) \leq \log_\gamma (1 - (1-\gamma)\frac{p}{b})$, where $0<\gamma< 1$ and $b>0$. 
\end{proposition}
\begin{proof} 
Collapse all $a$-states into a single state $t$ and make it absorbing (i.e., replace all outgoing transitions by a self-loop with probability one). 
Let $t$ be the only \emph{bottom strongly connected component} (BSCC) of $M$ (other BSCCs can be safely ignored). 
Let $\{1,\hdots,m\}$ be the states of the modified MC $M$, with the initial state $s_0$ and the target state $t$  represented by $1$ and $m$, respectively.
Let $\bf Q$ be the ($m{-}1)\times (m{-}1$) transition matrix of the modified MC without the state $t$. 
That is, ${\bf Q}(i,j) = {\bf P}(i,j)$ iff $j\neq m$ where ${\bf P}$ is the transition probability matrix of $M$.
We have the following observation:
\begin{enumerate}
\item 
Let the coefficient of ergodicity $\tau({\bf Q})$ of ${\bf Q}$ defined as 
$$ \tau({\bf Q}) \ = \ 1-\displaystyle\min_{i,j} \left(\sum_k \min\{{\bf Q}(i,k),{\bf Q}(j,k)\} \right).$$ 
As $\bf Q$ is sub-stochastic and no row of $\bf Q$ is zero, it follows $0<\tau({\bf Q})<1$. 
\item 
Let vector ${\bf r}^T=(r_1,\hdots,r_{m{-}1})$ with $r_i = {\bf P}(i,m)$, $r_{\max}$ be the maximum element in ${\bf r}$ and ${\bf i}^T$ be $(1,0,\hdots,0)$.
The probability of reaching the state $m$ from the state $1$ in at most $n{+}1$ steps is the probability of being in some state $i<m$ within $n$ steps and taking the next transition to $m$: \vspace{-0.2cm}
$$ \mu_{n+1}  =  \sum_{j=0}^{n+1} {\bf i}^T{\bf Q}^j{\bf r} \ \leq \ \sum_{j=0}^{n+1} \tau({\bf Q})^j r_{\max}.$$
\end{enumerate}
Let $\tau({\bf Q}) = \gamma$ and $r_{\max}=b$. 
The integer $n_0$ is the smallest integer such that $\mu_{n_0} \geq p$, which implies
that $b{\cdot}\frac{1-\gamma^{n_0}}{1-\gamma} \geq p$.
This yields $n_0\leq \log_\gamma (1 - (1-\gamma)\frac{p}{b})$. \hfill $\quad\blacksquare$ 
\end{proof}

%
%

As in the non-parametric setting, it follows that (for finite MCs) the valuation sets $V_{>0}(\de_{x} \, a)$ and $V_{=1}(\de_{x} \, a)$ can be determined by a graph analysis, i.e. no inspection of the transition probabilities is necessary for qualitative parametric reachability properties.
\begin{proposition}
 The problem $V_{>0}(\de_x \, a)=\emptyset$ is \emph{NL}-complete.
\end{proposition}
\begin{proof}
 The problem is the same as reachability in directed graphs. \hfill $\quad\blacksquare$ 
\end{proof}
\begin{proposition}
The sets $V_{>0}(\de_{x} \, a)$ and $V_{=1}(\de_{x} \, a)$ can be determined in polynomial time by a graph analysis of MC $M$.
\end{proposition}
\begin{proof}
Collapse all the $a$-states into a target state $t$ and make $t$ absorbing.
If $V_{>0}(\de_{x} \, a)$ is non-empty, it suffices to determine $\min V_{>0}(\de_{x} \, a)$ which equals the length of a shortest path from $s_0$ to $t$.
To determine whether $V_{=1}(\de_{x} \, a)$ is empty or not, we proceed as follows.
If a cycle without $t$ is reachable from $s_0$, then no finite $n$ exists for which the probability of reaching $t$ within $n$ steps equals one. 
Thus, $V_{=1}(\de_x \, a) = \emptyset$.
If this is not the case, then the graph of $M$ is a DAG (apart from the self-loop at $t$), and $\min V_{=1}(\de_x \, a)$ equals the length of a 
longest path from $s_0$ to $t$. \hfill $\quad\blacksquare$
\end{proof}

\section{The Fragment pLTL(F,X)}
\label{sec:pltlfx}

This section considers the fragment pLTL(F,X) which is defined by:
$$
\varphi \ ::= \ a \where \neg a \where \varphi \, \wedge \, \varphi \where \varphi \, \vee \, \varphi \where \X\! \varphi \where \de \varphi \where \de_{\leq x} \, \varphi \where \de_{\leq c} \, \varphi
$$ 
Our first result is a necessary and sufficient condition for the emptiness of $V_{{>} 0}(\varphi)$.

\begin{theorem}\label{FX}
For $\varphi \in$ pLTL(F,X) and MC $M$ with $m$ states, 
$V_{{>} 0}(\varphi) \neq \emptyset \mbox{ iff } \bar{{\a v}} \in V_{{>} 0}(\varphi)$ with $\bar{{\a v}}(x) = m{\cdot}|\varphi|$.
\end{theorem}
\begin{proof} 
Let $\varphi$ be a pLTL(F,X)-formula and assume $V_{> 0}(\varphi) \neq \emptyset$.
By monotonicity, it suffices to prove that ${\a v}\in V_{>0}(\varphi)$ with ${\a v} \not\leq \bar{{\a v}}$ implies $\bar{{\a v}} \in V_{>0}(\varphi)$.
The proof proceeds in a number of steps. 
(1) We show that it suffices to consider formulas without disjunction. 
(2) We show that if path fragment
$\pi[0..l] \models \bar\varphi$, (where LTL(F,X)-formula $\bar\varphi$ is obtained from $\varphi$ by omitting all parameters from $\varphi$) 
then $\pi[0..l] \models {\a v}_l(\varphi)$  with ${\a v}_l(x) = l$ for every $x$.
(3) We construct a deterministic B\"uchi automaton (DBA) $A_{\bar\varphi}$ for $\bar\varphi$ such that its initial and final state are at most $| \bar\varphi |$ 
transitions apart.
(4) We show that reachability of a final state in the product of MC $M$ and DBA $A_{\bar\varphi}$ implies the existence of a finite path in $M$  
of length at most $m{\cdot}|\varphi|$ satisfying $\bar\varphi$.
See the appendix for details. \hfill $\blacksquare$
\end{proof}
 The above Theorem \ref{FX} leads to the following proposition.
\begin{proposition}\label{np}
For $\varphi \in$ pLTL(F,X), deciding if $V_{>0}(\varphi)=\emptyset$  is NP-complete. 
\end{proposition}
\begin{proof}
 See the appendix.  
 \hfill $\blacksquare$
\end{proof}

\noindent
For almost sure properties, a similar approach as for $V_{>0}(\varphi)$ suffices.
\begin{theorem}
$\!\!\!\!$ For $\varphi \in$ pLTL(F,X) and MC $M$ with $m$ states,
$V_{{=} 1}(\varphi) \neq \emptyset \mbox{ iff } \bar{ {\a v}} \in V_{{=} 1}(\bar\varphi)$ 
with $\bar{ {\a v}}(x) = m{\cdot}|\varphi|$.
\end{theorem}
\begin{proof}
 Consider the direction from left to right.
The argument goes along similar lines as the proof of Theorem~\ref{FX}. We build the DBA $A_{\bar\varphi}$ for $\bar\varphi$ and take the cross product with Markov chain $M$. There
are $m{\cdot}|\varphi|$ state in the cross product. If $\Pr(M\models {\bar{\ve}}(\varphi))<1$ then there is some cycle in the cross product that
does not contain the final state. Thus, $V_{=1}(\varphi)$ is empty. \hfill $\blacksquare$
\end{proof}

\noindent
Theorem~\ref{FX} suggests that $\min V_{>0}(\varphi)$ lies in the hyper-cube $H = \set{0,\hdots,N}^d$, where $N=m{\cdot}|\varphi|$. 
A possible way to find $\min V_{>0}(\varphi)$ is to apply the bisection method in $d$-dimensions.
We recursively choose a middle point of the cube, say ${\a v} \in H$ ---in the first iteration ${\a v}(x)=N/2$--- 
and divide $H$ in $2^d$ equally sized hypercubes. 
If ${{\a v}} \in V_{> 0}(\varphi)$, then the hypercube whose points exceed ${\a v}$ is discarded, else the cube whose points are below 
${\a v}$ is discarded. 
The asymptotic time-complexity of this procedure is given by the recurrence relation:
\begin{equation}\label{re}
T(k) = (2^d-1)\cdot T(k{\cdot}2^{{-}d}) + F
\end{equation} 
where $k$ is the number of points in the hypercube and $F$ is the complexity of checking ${\a v} \in V_{> 0}(\varphi)$ 
where $|\a v|\leq N$. 
Section \ref{sec:algo} presents an algorithm working in ${\cal O}(m{\cdot}N^{d}{\cdot}2^{|\varphi|})$ for a somewhat more expressive logic. 
From (\ref{re}), this yields a complexity of ${\cal O}(m{\cdot}N^d{\cdot}2^{|\varphi|}{\cdot}\log N )$. 
The size of a set of minimal points can be exponential in the number of variables, as shown below.

\begin{proposition}\label{maxbound}
$|\min V_{>0}(\varphi)| \leq  (N{\cdot}d)^{d-1}$.
\end{proposition}
\begin{proof}   See the appendix.
\hfill $\blacksquare$\end{proof} 
\begin{figure}
\scalebox{0.7}{
\begin{tikzpicture}
 \foreach \x in {1,...,4}{
  \draw (\x,0) circle (2pt); \draw (\x+4,0) circle (2pt);
  
  \filldraw[red] (\x,4-\x) circle (2pt); \node at (\x,4.3-\x) {$r$}; \draw (\x+1,4-\x) circle (2pt); \path[->] (\x+0.1,4-\x) edge (\x+0.9,4-\x);

  \filldraw[blue] (\x+4,4-\x) circle (2pt); \node at (\x+4,4.3-\x) {$b$}; \draw (\x+1+4,4-\x) circle (2pt); \path[->] (\x+0.1+4,4-\x) edge (\x+0.9+4,4-\x);

  \ifthenelse{\x<4}{ \path[->] (\x+1.1,3.9-\x) edge (\x+1.9,3.1-\x); \draw (\x+1.5,3.5-\x) circle (2pt);
             \path[->] (\x,0.1) edge (\x,3.9-\x); \path[->] (\x+0.1,0) edge (\x+0.9,0); \path[->] (\x+4.1,0) edge (\x+4.9,0);

             \path[->] (\x+1+4.1,3.9-\x) edge (\x+1.9+4,3.1-\x); \draw (\x+4+1.5,3.5-\x) circle (2pt);
             \path[->] (\x+4,0.1) edge (\x+4,3.9-\x); \path[->] (\x+0.1+4,0) edge (\x+0.9+4,0); \path[->] (\x+4+4.1,0) edge (\x+4.9+4,0);
 
             \draw (9+\x,0) circle (2pt);
            }{} 
}
\filldraw[red] (4,0) circle (2pt); \draw (3.5,0) circle (2pt);
\filldraw[blue] (8,0) circle (2pt); \draw (7.5,0) circle (2pt);
\filldraw[green] (12,0) circle (2pt); \node at (12,0.3) {$g$};

\node at (11,4) {$x_1$}; \node at (12,4) {$x_2$}; \node at (13,4) {$x_3$};
\draw (10.5,3.7) -- (13.5,3.7);
\draw (11.5,4.2) -- (11.5,0.9);
\draw (12.5,4.2) -- (12.5,0.9);

\node at (11,3.4) {\footnotesize{5}}; \node at (12,3.4) {\footnotesize 10}; \node at (13,3.4) {\footnotesize 14};
\node at (11,3.1) {\footnotesize{5}}; \node at (12,3.1) {\footnotesize 9}; \node at (13,3.1) {\footnotesize 15};
\node at (11,2.8) {\footnotesize{5}}; \node at (12,2.8) {\footnotesize 8}; \node at (13,2.8) {\footnotesize 16};
\node at (11,2.5) {\footnotesize{5}}; \node at (12,2.5) {\footnotesize 7}; \node at (13,2.5) {\footnotesize 17};

\node at (11,2) {\footnotesize{4}}; \node at (12,2) {\footnotesize 11}; \node at (13,2) {\footnotesize 15};
\node at (11,1.7) {\footnotesize{4}}; \node at (12,1.7) {\footnotesize 10}; \node at (13,1.7) {\footnotesize 16};
\node at (11,1.4) {\footnotesize{4}}; \node at (12,1.4) {\footnotesize 9}; \node at (13,1.4) {\footnotesize 17};
\node at (11,1.1) {\footnotesize{4}}; \node at (12,1.1) {\footnotesize 8}; \node at (13,1.1) {\footnotesize 18};

\node at (14.5,4) {$x_1$}; \node at (15.5,4) {$x_2$}; \node at (16.5,4) {$x_3$};
\draw (14,3.7) -- (17,3.7);
\draw (15,4.2) -- (15,0.9);
\draw (16,4.2) -- (16,0.9);

\node at (14.5,3.4) {\footnotesize{3}}; \node at (15.5,3.4) {\footnotesize 10}; \node at (16.5,3.4) {\footnotesize 16};
\node at (14.5,3.1) {\footnotesize{3}}; \node at (15.5,3.1) {\footnotesize 11}; \node at (16.5,3.1) {\footnotesize 17};
\node at (14.5,2.8) {\footnotesize{3}}; \node at (15.5,2.8) {\footnotesize 10}; \node at (16.5,2.8) {\footnotesize 18};
\node at (14.5,2.5) {\footnotesize{3}}; \node at (15.5,2.5) {\footnotesize 9}; \node at (16.5,2.5) {\footnotesize 19};

\node at (14.5,2) {\footnotesize{2}}; \node at (15.5,2) {\footnotesize 13}; \node at (16.5,2) {\footnotesize 17};
\node at (14.5,1.7) {\footnotesize{2}}; \node at (15.5,1.7) {\footnotesize 12}; \node at (16.5,1.7) {\footnotesize 18};
\node at (14.5,1.4) {\footnotesize{2}}; \node at (15.5,1.4) {\footnotesize 11}; \node at (16.5,1.4) {\footnotesize 19};
\node at (14.5,1.1) {\footnotesize{2}}; \node at (15.5,1.1) {\footnotesize 10}; \node at (16.5,1.1) {\footnotesize 20};
\end{tikzpicture}}
\vspace{-0.5cm}\\
\caption{MC and $\min V_{> 0}(\varphi)$ for pLTL(F,X)-formula $\varphi = \de_{x_1} \, \mbox{\sl r} \wedge \de_{x_2} \, \mbox{\sl b} \wedge \de_{x_3} \, \mbox{\sl g}$}

\label{maxfig}
 \end{figure} 
\begin{example}\label{maxex}
There exist MCs for which $| \min V_{>0}(\varphi) |$ grows exponentially in $d$, the number of parameters in $\varphi$, whereas the number $m$ of states in the MC grows linearly in $d$.
For instance, consider the MC $M$ in Fig.~\ref{maxfig} and $\varphi = \de_{x_1} \, r \wedge \de_{x_2} \, b \wedge \de_{x_3} \, g$, i.e., $d{=}3$.
We have $| \min V_{>0}(\varphi) | = 4^2$ as indicated in the table.
\end{example}

\noindent 
We conclude this section by briefly considering the membership query: does ${ {\a v}} \in V_{> 0}(\varphi)$ for pLTL(F,X)-formula 
$\varphi$ with $d$ parameters?
Checking membership of a valuation $\ve\in V_{>0}(\varphi)$ boils down to deciding whether there exists a $\ve'\in \min V_{>0}(\varphi)$ such that $\ve \geq \ve'$.
A representation of $\min V_{> 0}(\varphi)$ facilitating an efficient membership test can be obtained by putting all elements in this set in lexicographical order.
This involves sorting over all $d$ coordinates.
A membership query then amounts to a recursive binary search over $d$ dimensions.
This yields:

\begin{proposition}
For pLTL(F,X)-formula $\varphi$, ${ {\a v}} \in V_{>0}(\varphi)?$ takes ${\cal O}(d{\cdot}\log N{\cdot}d)$ time, provided a representation of $\min V_{>0}(\varphi)$ is given.
\end{proposition}

\section{Qualitative Parametric B\"uchi}
\label{sec:parambuechi}

In this section, we consider pLTL-formulas of the form $\varphi = \Box\de_x \, a$, for proposition $a$. 
We are interested in $V_{>0}(\varphi)$, i.e., does the set of infinite paths visiting $a$-states that are maximally $x$ apart infinitely often, have a positive measure?
Let MC $M = (S, \bfP, s_0, L)$. 
A \emph{bottom strongly-connected component} (BSCC) $B \subseteq S$ of $M$ is a set of mutually reachable states with no edge leaving $B$. 
For BSCC $B$, let $n_{a,B} = \max\set{ |\pi| \mid \forall i \leq |\pi|, \pi[i] \in B \wedge a \notin L(\pi[i])}$. 
%

\begin{proposition}\label{buchi}
Let $B$ be a BSCC and $s \in B$.  Then, $\forall n\in\mathbb{N}, n > n_{a,B} \Leftrightarrow \Pr(s\models \Box\de_n\ a) =1$ and $n \leq n_{a,B} \Leftrightarrow \Pr(s\models \Box\de_n\ a) =0$.
\end{proposition}
\begin{proof}
If $n > n_{a,B}$, then each path $\pi$ from any state $s \in B$ will have at least one $a$-state in finite path fragment $\pi[i,\hdots,i{+}n]$ for all $i$. 
Hence, $\Pr(s\models\Box\de_n \ a) =1$. 
If $n\leq n_{a,B}$, then there exists a finite path fragment $\rho$ of $B$, such that, for all $i \leq n$, $a \notin L(\rho[i])$. 
Consider an infinite path $\pi$ starting from any arbitrary $s \in B$. 
As $s \in B$, $\pi$ will almost surely infinitely often visit the initial state of $\rho$.
Therefore, by~\cite[Th.10.25]{DBLP:books/daglib/0020348}, $\pi$ will almost surely visit every finite path fragment starting in that state, in particular $\rho$.
Path $\pi$ thus almost surely refutes $\Box\de_n\  a$, i.e. $\Pr(s\models \Box\de_n\ a) =0$.
\hfill $\blacksquare$ \end{proof}

\noindent
For any BSCC $B$ and $\Box\de_x\ a$, $n_{a,B}<\infty$ iff every cycle in $B$ has at least one $a$-state.
Hence, $n_{a,B}$ can be obtained by analysing the digraph of $B$ (in ${\cal O}(m^2)$, the number of edges).
BSCC $B$ is called \emph{accepting} for $\Box\de_x \, a$ if $n_{a,B}<\infty$ and $B$ is reachable from the initial state $s_0$. 
Note that this may differ from being an accepting BSCC for $\Box\de a$.
Evidently, $V_{>0}(\pmb\Box\de_x \,a) \neq \emptyset$ iff $n_{a,B} < \infty$. 
This result can be extended to \emph{generalized B\"uchi} formula
$\varphi = \Box\de_{x_1} \, a_1 \wedge \hdots \wedge \Box\de_{x_d} \, a_d$, by checking $n_{a_i,B}< \infty$ for each $a_i$. 

As a next problem, we determine $\min V_{>0}(\Box\de_{x} \, a)$. For the sake of simplicity, let MS $M$ have a single
accepting BSCC $B$.
For states $s$ and $t$ in MC $M$, let $d(s,t)$ be the distance from $s$ to $t$ in the graph of $M$.
(Recall, the distance between state $s$ and $t$ is the length of the shortest path from $s$ to $t$.)
For BSCC $B$, let $d_{a,B}(s) = \min_{t \in B, a \in L(t)} d(s,t)$, i.e., the minimal distance from $s$ to an $a$-state in $B$.
Let the proposition $a_B$ hold in state $s$ iff $s \in B$ and $a \in L(s)$.
Let $G_a = (V,E)$ be the digraph defined as follows: $V$ contains all $a$-states of $M$ and the initial state $s_0$ and $(s,s') \in E$ iff there is path from $s$ to $s'$ in $M$.
Let $c$ be a cost function defined on a finite path $s_0 \ldots s_n$ in graph $G_a$ as: $c(s_0 \ldots s_n) = \max_{i} d(s_i,s_{i+1})$, 
($d$ is defined on the graph of $M$). 
Using these auxiliary notions we obtain the following characterization for $\min V_{>0}(\Box\de_{x} \, a)$:

%
%
%
%

\begin{theorem}\label{buchimain}
$\min V_{>0}(\Box\de_x \, a) = n_0$ where $n_0 = \displaystyle\max \tiny{\left( n_{a,B}, \min_{\pi = s_0 \ldots s_n, s_n \models a_B} c(\pi) \right)}$ 
if $n_{a,B} < d_{a,B}(s_0)$ and $n_0 = n_{a,B}$ otherwise.
\end{theorem}
\begin{proof}
See the appendix. \hfill $\blacksquare$
 \end{proof}

\noindent
If MC $M$ has more than one accepting BSCC, say $\set{B_1, \ldots, B_k}$ with $k > 1$, then $n_0 = \min_i n_{0,B_i}$, where $n_{0,B_i}$  for $0 < i \leq k$ is obtained as in Theorem \ref{buchimain}.

\begin{proposition}
The sets $V_{>0}(\Box \de_{x} \, a)$ and $V_{=1}(\Box \de_{x} \, a)$ can be determined in polynomial time by a graph analysis of MC $M$.
\end{proposition} 
\begin{proof} 
See the appendix. \hfill $\blacksquare$
\end{proof}

\noindent
Determining $\min V_{\geq p}(\Box \de_{x} \, a)$ for arbitrary $p$ reduces to reachability of \emph{accepting} BSCCs.
In a similar way as for parametric reachability (cf. Section \ref{sec:reach}), this  can be done searching.  
For generalized B\"uchi formula $\varphi = \bo \de_{x_i} \, a_i \wedge\hdots\wedge \bo\de_{x_d} \, a_d$ and BSCC $B$, $n_{a_i B}$ is at most $m$. 
Thus, $\min V_{>0}(\varphi) \in \set{0,\hdots,m{\cdot}d}^d$ and can be found by the bisection method, 
similar to the procedure described in Section \ref{sec:pltlfx}. 

\section{The Fragment pLTL$_\de$}
\label{sec:nobox}

This section is concerned with the logical fragment pLTL$_\de$, as defined in (\ref{pltl_d}):
$$
\varphi \ ::= \ a \where \neg a \where \varphi \, \wedge \, \varphi \where \varphi \, \vee \, \varphi \where \X\! \varphi \where \varphi \U \varphi 
\where \varphi \Release \varphi 
\where \Box \varphi \where \de_{\leq x} \, \varphi. \footnote{The modalities  $\de_{\leq c}$ and $\Box_{\leq c}$ can be removed with only quadratic blow up.}
$$
We will focus on the emptiness problem: is $V_{>0}(\varphi) = \emptyset$.
The decision problem whether $V_{=1}(\varphi)$ is very similar.
Similar as for pLTL(F,X), we obtain necessary and sufficient criteria for both cases.
The proofs for these criteria depend on an algorithm that checks whether ${{\a v}} \in V_{>0}(\varphi)$.
This algorithm is presented first.

\paragraph{Automata constructions.}
\label{sec:algo}

Let $\varphi$ be a pLTL$_\de$-formula, and ${\a v}$ a variable valuation.
W.l.o.g. we assume that each variable occurs once in $\varphi$.
We will extend the classical automaton-based approach for LTL by constructing a nondeterministic B\"uchi automaton for $\varphi$ that is amenable to treat the variables occurring in $\varphi$.
To that end, inspired by \cite{DBLP:journals/tcs/Zimmermann13}, we proceed in a number of steps:
\begin{enumerate}
\item 
Construct an automaton $G_\varphi$ for $\varphi$, independent from the valuation ${\a v}$, with two types of acceptance sets, one for treating until and release-modalities (as standard for 
LTL~\cite{Vardi96anautomata-theoretic}), and one for treating the parameter constraints.
\item 
Establish how for a given valuation ${\a v}$, a B\"uchi automaton $B_\varphi({\a v})$ can be obtained from $G_\varphi$ such that for infinite word $w$, $(w, {\a v}) \in \L(\varphi)$ iff $w$ is an accepting run of $B_\varphi({\a v})$.
\item 
Exploit the technique advocated by Couvreur \emph{et al.} \cite{CouvreurSahebSutre03} to verify MC $M$ versus $B_\varphi(\ve)$.
\end{enumerate}
We start with constructing $G_\varphi$.
Like for the LTL-approach, the first step is to consider consistent sets of sub-formulas of $\varphi$.
Let $\cl(\varphi)$ be the set of all sub-formulas of $\varphi$. 
Set $H \subseteq \cl(\varphi)$ is \emph{consistent}, when: \\[1ex] 
\begin{minipage}{0.5\textwidth}
\begin{itemize}
    \item $a\in H$ iff $\neg a \not\in H$,
    \item $\varphi_1 \wedge \varphi_2\in H$ iff $\varphi_1 \in H \mbox{ and } \varphi_2\in H$,
    \item $\varphi_1 \vee \varphi_2 \in H$ iff $\varphi_1\in H \mbox{ or } \varphi_2\in H$,
\end{itemize}
\end{minipage}
\begin{minipage}{0.5\textwidth}
\begin{itemize}
    \item $\varphi_2 \in H$ implies $\varphi_1 \U\varphi_2 \in H$,
    \item $\varphi_1,\varphi_2 \in H$ implies $\varphi_1 \Release \varphi_2 \in H$,
    \item $\varphi_1 \in H$ implies $\de_{x} \, \varphi_1 \in H$.
\end{itemize}
\end{minipage} \\[1ex]
\noindent
We are now in a position to define $G_\varphi$, an automaton with two acceptance sets.
For $\varphi \in$ pLTL$_\de$, let $G_\varphi= (Q,2^{AP},Q_0,\delta,\Acc_B,\Acc_P)$ where
 \begin{itemize}
  \item $Q$ is the set of all consistent sub-sets of  $\cl(\varphi)$  and $Q_0 = \set{H \in Q \mid \varphi \in H}$.
  \item $(H,a,H')\in \delta$, where $a \in 2^{AP}$ whenever:
  \begin{itemize}
   \item $H \cap AP = \set{a}$,
   \item $\X\! \varphi_1\in H \iff \varphi_1 \in H'$,
   \item $\varphi_1\U\varphi_2 \in H \iff \varphi_2 \in H$ or $(\varphi_1\in H \mbox{ and } \varphi_1\U\varphi_2 \in H')$,
   \item $\varphi_1\Release\varphi_2 \in H \iff \varphi_2 \in H$ and $(\varphi_1\in H \mbox{ or } \varphi_1\Release\varphi_2 \in H')$,
   \item $\de_{x} \, \varphi_1 \in H \iff \varphi_1 \in H$ or $\de_{ x} \, \varphi_1 \in H'$,
  \end{itemize}
  \item (generalized) B\"uchi acceptance $\Acc_B$ and parametric acceptance $\Acc_P$:
  \begin{itemize}
   \item $\Acc_B = \set{F_{\varphi'} \mid \varphi' \in \cl(\varphi) \wedge (\varphi' = \varphi_1{\U}\varphi_2 
             \vee \varphi' = \varphi_1{\R}\varphi_2)}$ where
      \begin{itemize}
      \item 
      $F_{\varphi'} = \set{H \mid \varphi' \in H \Rightarrow \varphi_2 \in H}$ if $\varphi' = \varphi_1 \U \varphi_2$, and
      \item
      $F_{\varphi'} = \set{H \mid \varphi_2 \in H \Rightarrow \varphi' \in H}$ if $\varphi' = \varphi_1 \R \varphi_2$,
      \end{itemize}
   \item $\Acc_P = \set{F_{{x_i}} \, | \, \de_{x_i} \, \varphi_i \in \cl(\varphi)}$ with $F_{{x_i}}  = \set{H \, | \, \de_{x_i} \, \varphi_i \in H \Rightarrow \varphi_i \in H}$. 
  \end{itemize}
 \end{itemize}
  A run $\rho \in Q^\omega$ of $G_\varphi$ 
  is \emph{accepting} under valuation $\ve$ if it visits each set in $\Acc_B$ infinitely often and each $F_{x_i} \in \Acc_P$ 
  in every infix of length $\ve(x_i)$.
  $\L(G_\varphi )$ contains all pairs $(w,\ve)$ such that there is an accepting run of $w$ under the valuation $\ve$. 
$G_\varphi$ is \emph{unambiguous} if $q \xrightarrow{a} q'$ and $q\xrightarrow{a}q''$ implies $\L(q') \cap \ \L(q'')=\emptyset$, where $\L(q)$ is the language starting from the state $q$.
\begin{proposition}[\cite{DBLP:journals/tcs/Zimmermann13}]
For $\varphi \in$ pLTL$_\de$, the automaton $G_\varphi $ is unambiguous and $\L(G_\varphi ) = \L(\varphi)$.
\end{proposition}

\noindent
The automaton $G_\varphi $ can be constructed in ${\cal O}(2^{|\varphi|})$. 
Apart from the parametric acceptance condition, $G_\varphi $ behaves as a generalized B\"uchi automaton (GNBA) with accepting set $\Acc_B = \set{F_1,\ldots,F_k}$.
In order to obtain a non-deterministic automaton, we first apply a similar transformation as for GNBA to NBA~\cite{DBLP:books/daglib/0020348}. 
We convert $G_\varphi $ to $U_\varphi = (Q',2^{\smAP},Q_0',\delta',\Acc_B',\Acc_P')$ where $Q' = Q \times \set{1, \ldots, k}$, $Q_0' = Q_0\times \set{1}$.
If $(q,a,q') \in \delta$, then $((q,i),a,(q',i')) \in \delta'$ with $i{=}i'$ if $q\not\in F_i$ else $i'= (i \mod k){+}1$. 
$\Acc_B= F_1\times\set{1}$ and $\Acc_P'= \set{F_{x_i}' \mid F_{x_i} \in \Acc_P}$, where $F_{x_i}'= F_{x_i} \times \set{1,\ldots,k}$. 
Note that the construction preserves unambiguity and the size of $U_\varphi$ is in ${\cal O}(|\varphi|{\cdot}2^{|\varphi|})$.
%

For a given valuation ${\a v}$, $U_\varphi$ can be converted into an NBA $B_\varphi({\a v})$.
This is done as follows.
Let $U_\varphi =(Q',2^{AP},Q_0',\delta',\Acc'_B,\Acc'_P)$ and ${\a v}$ a valuation of $\varphi$ with $d$ parameters.
Then $B_\varphi({\a v}) = (Q'',2^{AP},Q''_0,\delta'',\Acc)$ with:
 \begin{itemize}
  \item $Q'' \subseteq Q'\times \{0,\hdots,{\a v}(x_1)\} \times \hdots \times \{0,\hdots,{\a v}(x_d)\}$,
  \item $((q,\mathbf n),a,(q',\mathbf n')) \in \delta''$ if $(q,a,q')\in \delta'$ and for all $x_i$:
  \begin{itemize}
   \item if $q'\in F'_{x_i}$ and $\mathbf n(x_i)< {\a v}(x_i)$ then $\mathbf n'(x_i)=0$,
   \item if $q'\notin F'_{x_i}$ and $\mathbf n(x_i) <{\a v}(x_i)$ then $\mathbf n'(x_i)= \mathbf n(x_i)+1$.
  \end{itemize}
  \item $Q_0'' = Q_0'\times {0}^d$ and $\Acc =\Acc'_B \times\{0,\hdots,{\a v}(x_1)\} \times \hdots \times \{0,\hdots,{\a v}(x_d)\}$.
 \end{itemize}
It follows that $B_\varphi(\ve)$ is unambiguous for any valuation $\ve$. 
Furthermore, every run of $B_\varphi({\a v})$ is either finite or satisfies the parametric acceptance condition for valuation ${\a v}$. 
Thus we have:
\begin{proposition}
An infinite word $w \in \L(B_\varphi({\a v}))$ if and only if $(w,{\a v})\in \L(\varphi)$.
\end{proposition}
The size of $B_\varphi(\ve)$ is in ${\cal O}(c_{{\a v}}{\cdot}|\varphi|{\cdot}2^{|\varphi|})$ where $c_{{\a v}} = \prod_{x_i}\mathbf(\ve(x_i)+1)$.

As a next step, we exploit the fact that $B_\varphi(\ve)$ is unambiguous, and apply the technique by Couvreur \emph{et al.}~\cite{CouvreurSahebSutre03} for verifying MC $M$ against $B_\varphi(\ve)$.
Let $M\otimes B_\varphi({\a v})$ be the synchronous product of $M$ and $B_\varphi({\a v})$~\cite{DBLP:books/daglib/0020348}, $\Pi_1$ the projection to $M$ and $\Pi_2$ the projection to $B_\varphi(\ve)$. 
Let $\L(s,q) =  \set{ \pi \in\Paths(s) \mid \trace(\pi) \in \L(q)}$ and $\Pr(s,q) = \Pr(\L(s,q))$. 
Let $\Pr(M \otimes B_\varphi(\ve)) = \sum_{q_0 \in Q_0} \Pr(s_0,q_0)$.
As $B_\varphi({\a v})$ is unambiguous, we have for any $(s,q)$:
$$
\Pr(s,q) \ = \ \sum_{(t,q')\in \delta(s,q)} \bfP(s,t) \cdot \Pr(t,q'),
$$ 
where $\delta$ is the transition relation of $M\otimes B_\varphi({\a v})$ and $\bfP(s,t)$ is the one-step transition probability from $s$ to $t$ in MC $M$. 
A (maximal) strongly connected component (SCC, for short) $C \subseteq S$ is \emph{complete} if for any $s \in \Pi_1(C)$ :
$$
\Paths(s) = \bigcup_{(s,q) \in C} \L_C(s,q)
$$ 
where $\L_C(s,q)$ \emph{restricts} runs to $C$ (runs only visits states from $C$). 
The SCC $C$ is \emph{accepting} if $\Acc \,\cap \Pi_2(C)\neq \emptyset$ (where $\Acc$ is the set of accepting states in $B_\varphi(\ve)$). 

\begin{proposition}[\cite{CouvreurSahebSutre03}\label{accepting-bscc}]
Let $C$ be a complete and accepting SCC in $M\otimes B_\varphi(\ve)$. 
Then for all $s\in \Pi_1(C)$:
$$ 
\Pr\bigg(\bigcup_{(s,q)\in C} \L_C(s,q)\bigg) = 1.
$$ 
Moreover, since $B_\varphi({\a v})$ is unambiguous, $\Pr(M\otimes B_\varphi({\a v})) >0$ implies there exists a reachable, complete and accepting SCC. 
\end{proposition}

Finding complete and accepting SCC in $M\otimes B_\varphi(\ve)$ is done by standard graph analysis. 
Altogether, ${{\a v}} \in V_{>0}(\varphi)$ is decided in ${\cal O}(m{\cdot}c_{{\a v}}{\cdot}|\varphi|{\cdot}2^{|\varphi|})$. 
The space complexity is polynomial in the size of the input (including the valuation), as $M\otimes B_\varphi({\a v})$ can be stored in $\mathcal{O}(\log m + |\varphi| + \log c_{{\a v}})$ bits.
In the sequel, we exploit these results to obtain a necessary and sufficient criterion for the emptiness of $V_{> 0}(\varphi)$ for $\varphi$ in  pLTL$_\de$.


\begin{theorem}\label{th:pltld} 
$\!\!\!\!$ 
For $\varphi \in$ pLTL$_\de$, 
$V_{> 0}(\varphi) \neq \emptyset \mbox{ iff } \bar{{\a v}} \in V_{> 0}(\varphi)$ s.t. $\bar{{\a v}}(x) = m{\cdot}|\varphi|{\cdot}2^{|\varphi|}$.
\end{theorem}
\begin{proof} 
Consider the direction from left to right.
The only non-trivial case is when there exists a valuation ${\a v} \not\leq \bar{{\a v}}$ such that ${\a v}\in V_{>0}(\varphi)$ implies $\bar{{\a v}}\in V_{>0}(\varphi)$. 
In the model checking algorithm described above, we first construct $G_\varphi$, and then $U_\varphi$ with a single B\"uchi accepting set $\Acc'_B$ and $d$ parametric accepting sets $F'_{x_i}$, one for each variable $x_i$ in $\varphi$. 
For the sake of clarity, assume $d=1$, i.e., we consider valuation $v$.
The explanation extends to the general case in a straightforward manner.
For valuation $v$, consider $M \otimes B_\varphi(v)$. 
We show that, for $r < v$, $\Pr(M \otimes B_\varphi(v)) >0$ implies $\Pr(M\otimes B_\varphi(r)) >0$,
where $r = m{\cdot}|U_\varphi|$, which is in ${\cal O}(m{\cdot}|\varphi|{\cdot}2^{|\varphi|})$. 

Note that every cycle in $M \otimes B_\varphi(r)$ contains a state $(s,q,i)$ with $i=0$.
Moreover, the graph of $M \otimes B_\varphi(r)$ is a sub-graph of $M\otimes B_\varphi(v)$. 
We now prove that, if a (maximal) SCC $C$ of $M\otimes B_\varphi(r)$
is not complete (or accepting) then any SCC $C'$ of $M\otimes B_\varphi(v)$ containing $C$ is also not complete (or accepting, respectively).   

(a) Suppose $C$ is not complete. 
Then there exists a finite path $\sigma=s \, s_1 \ldots s_k$ of $M$, such that from any $q$, with $(s,q,0) \in C$, the run $\rho =(s,q,0)(s_1,q_1,1) \ldots (s_j,$ $q_j,j)$ leads to a deadlock state. 
This can have two causes: either $(s_j,q_j,j)$ has no successor for any $j$. 
Then, $C'$ is not complete.
Or, the path $\rho$ terminates in $(s_j,q_j,j)$ where $j=r$.
This means, for all $(s',q',j{+}1) \in \delta(s_j,q_j,j)$ in $C'$, $q' \not\in F_x$. 
As the length of $\rho$ exceeds $r$, there are states in the run whose first and second component appear multiple times.
Thus, we can find another path $\sigma'$ (possibly longer than $\sigma$) for $C'$ which goes through states where the first and the second component of some of its states are repeated sufficiently many times to have a run $(s,q,0)(s_1,q_1,1)\ldots(s_j,q_j,v)$ which is a deadlock state. 
Thus, $C'$ is not complete.
 
(b) Suppose $C'$ is accepting. 
Then there exists $(s',q',i')$ with $q'\in \Acc$. 
Since $C'$ is an SCC and $C\subseteq C'$, there is a path from $(s,q,0)$  $\in C$ to $(s',q',i')$. 
If the length of the path is less than $r$, then we are done. 
If $i'>r$, then some $(s'',q'')$ pair in the path must be repeated. 
Thus, we can find another path of length less than $r$ to a state $(s',q',i)$, where $i \leq r$. 
Therefore, $C$ is accepting. 
The rest of the proof follows from Proposition~\ref{accepting-bscc}.
\hfill $\blacksquare$
\end{proof}
For almost sure properties, a similar approach as for $V_{>0}(\varphi)$ suffices.

\begin{theorem}
$\!\!\!\!$
For $\varphi\in$ pLTL$_\de$, $V_{= 1}(\varphi) \neq \emptyset \mbox{ iff } \bar{ {\a v}} \in V_{= 1}(\bar\varphi)$ with $\bar{ {\a v}}(x) = m{\cdot}|\varphi|{\cdot}2^{|\varphi|}$.
\end{theorem}
\noindent
Let $N_{\varphi M} = m{\cdot}|\varphi|{\cdot}2^{|\varphi|}$.
Note that $c_{\bar \ve}$ equals $(N_{\varphi M})^d$. 
Thus, we have:

\begin{proposition}
 For $\varphi \in $ pLTL$_\de$, deciding if $V_{>0}(\varphi) = \emptyset$ is PSPACE-complete.
\end{proposition}
\begin{proof}
Theorem \ref{th:pltld} gives an algorithm in PSPACE, as $M\otimes B_\varphi(\bar\ve)$ can be stored in $O(\log m + |\varphi| + d\log N_{\varphi M})$ bits. 
PSPACE hardness follows trivially, as for LTL formula $\varphi$ and MC $M$, deciding $\Pr(M\models \varphi)> 0$ (which is known to be a PSPACE complete problem) is the same as checking the emptiness of $V_{>0}(\varphi)$.  \hfill $\blacksquare$
\end{proof}

Just as for pLTL(F,X), we can use the bisection method to find $\min V_{>0}(\varphi)$. The search procedure invokes the model checking algorithm
multiple times. We can reuse the space each time we check $\Pr(M\models\ve(\varphi))>0$.
Hence, $\min V_{>0}(\varphi)$ can be found in polynomial space. 
The time complexity of finding $\min V_{>0}(\varphi)$ is ${\cal O}(m{\cdot}(N_{\varphi M})^d{\cdot}2^{|\varphi|}{\cdot}\log N_{\varphi M})$. 
Membership can also be similarly solved.
\begin{proposition}
For pLTL$_\de$-formula $\varphi$, $\ve \in V_{>0}(\varphi)?$ takes ${\cal O}(d{\cdot} \log \frac{N_{\varphi M}}{d})$ time, provided a 
representation of $V_{>0}(\varphi)$ is given.  
\end{proposition}

\section{Concluding Remarks}
\label{sec:concl}

This paper considered the verification of finite MCs against parametric LTL. 
We obtained several results on the emptiness problem for qualitative verification problems, including
necessary and sufficient conditions as well as some complexity results.
Future work consists of devising more efficient algorithms for the quantitative verification problems, and lifting the results to extended temporal logics~\cite{Vardi94reasoningabout} and stochastic games, possibly exploiting~\cite{DBLP:journals/tcs/Zimmermann13}.


\paragraph*{Acknowledgement.}
This work was partially supported by the EU FP7 projects MoVeS and Sensation,  the EU Marie Curie project MEALS and the Excellence initiative of the German federal government.

\bibliographystyle{splncs}
 \bibliography{pltl} 
 \newpage
\appendix
 \section{} 
 \textbf{Proposition 1.} The problem $V_{>0} (\varphi)=\emptyset$ is undecidable for $\varphi \in$ pLTL.
 \begin{proof}
The proof is based on  \cite[Th.\ 4.1]{DBLP:journals/tocl/AlurETP01}, 
where the problem of deciding the existence of a halting computation of a two-counter machine is reduced to the satisfiability of $\varphi$. 

Let $T$ be a counter machine with two counters $\{c_1,c_2\}$ and $k+1$ states $\{s_0,s_1,\hdots,s_k\}$, $s_0$ being the initial state and $s_k$ the 
halting state. We construct a formula $\varphi_T$ such that any satisfiable structure $(w,\ve)$, represent a sequence of configurations of $T$ that constitute
a halting computation. In other words, the sequence of letters in the word $w$, will encode the an halting computation of $T$ for the valuation $\ve$.

 Crucial argument is that, a parameter  
can be used to guess the maximum value of each counter in any halting computation of $T$. Thus using a parameter say $x$, each configuration of
a halting computation can be stored in length $x$. We use propositions $\{p_1,\hdots,p_k\}$ for each state.
Let $b$ be $0$ (or $1$) and $\bar b$ be $1$ (or $0$, respectively).  

Word $w$ constains alternating sequence $\{q_0^0,q_1^0\}$ denoting the start and end of a configuration. The distance
between $q_0^0$, $q_0^1$ is exactly $x$. This is imposed by the formula:
\[
 \varphi_1 :=  \bigwedge_{b}\bigg(q_0^b \to \X(\neg q_0^{\bar b} \U\!_{=x}\ q_0^{\bar b}) \wedge \X(q_0^b\U q_0^{\bar b})\bigg).
\] The propositions $\{q_1^{-0},q_1^{0},q_1^{+0}\}$ (or $\{q_1^{-1},q_1^{1},q_1^{+1}\}$) will be used to keep track of counter $c_1$ in the configuration
starting with $q_0^0$ (or $q_0^1$, respectively). Similarly, $\{q_2^{-0},q_2^{0},q_2^{+0}\}$, $\{q_2^{-1},q_2^{1},q_2^{+1}\}$ do the same for counter $c_2$. We impose 
the condition that all these propositions occur exactly once between $q_0^0$ and $q_0^1$, and $\{q_i^{-b},q_i^{b},q_i^{+b}\}$ ($i=1,2$) are always
occur consecutively.
\begin{align*}
 \varphi^{bi}:= \bigg( (q_0^b\to \neg q_i^{\bar b}\U q^b_i)\wedge (q_i^b\to \X(\neg q_i^b\U q_0^b)) \wedge (q_i^{+b}\to\X(\neg q_i^{+b}\U q_0^{b}))\\
                \wedge\ (q_i^{-b}\to\X(\neg q_i^{-b}\U q_0^{b})) \wedge (q_i^b \to \X q_i^{+b}) \wedge (q_i^{-b} \to \X q_i^b)\bigg).
  \end{align*} 
Let $\varphi_2:= \bigwedge_{b,i=1,2}\varphi^{bi}$.  
Consider a configuration $(s_i,c_1,c_2)$ of $T$. If this configuration occur in a halting computation, then it is encoded in $w$ as a sub-sequence  of propositions (of length $x+2$) between $q_0^0$ and $q_0^1$.
Exactly one of the state proposition, $p_i$ in this case is true at the start of the configuration $q_0^b$. This is  imposed by,
\[
 \varphi_3 := \bigwedge_b\bigg(q_0^b\to P \wedge \X(P'\U q_0^b)\bigg)
\] where $P:= (p_1\wedge \neg p_2 \hdots \neg p_k) \vee \hdots \vee (p_k \wedge \neg p_1 \hdots \neg p_{k-1})$ and 
$P':= (\neg p_1\wedge\hdots\wedge \neg p_k)$. 
The distance of $q_1^b$ from $q_0^b$ will be used
to keep track of the value of $c_1$. To be precise, at a distance $c_1$ from $q_0^b$ the sequence $q_1^{-b},q_1^{b},q_1^{+b}$ occurs.

Consider a transition $e:s_i\xrightarrow{c_1:=c_1+1} s_j$. So if we are in a configuration where the distance of $q_1^{-b}$ from $q_0^b$ is $c_1$ then
in the next configuration, the distance of $q_1^{-\bar b}$ from $q_0^{\bar b}$ is $c_1+1$ or the distance of $q_1^{b}$ to $q_1^{-\bar b}$ is $x$. This
can be encoded as:
\[
 \varphi_e := \bigwedge_b\bigg( (q_0^b\wedge p_i) \to (\neg q_0^{\bar b}\U p_j \wedge \neg q_0^{\bar b}\U(q_1^{b}\to\X(\neg q_1^b\U_{=x}q_1^{-\bar b})))\bigg) 
\] A similar formula can be defined for transitions where the counter is decremented. For a transition where a counter value is compared to $0$,
$e: s_i\xrightarrow{c_1=0}s_j$ is encoded as:
\[
 \varphi_e := \bigwedge_b\bigg((q_0^b\wedge p_i) \to (\neg q_0^{\bar b}\U p_j \wedge \X q_1^{-b})\bigg).
\] Thus, the entire transition relation of $T$ can be encoded as $\varphi_4 := \bigvee_e \varphi_e$. 
\[
\varphi_T:= q_0^0\wedge (\bigwedge_{i=1}^4 \varphi_i) \U p_k.  
\] As a satisfiable structure of $\varphi_T$ encodes a halting computation of $T$ (vice-versa), 
satisfiability of $\varphi_T$ becomes undecidable. Furthermore,
 if $(w,\ve)$ satisfies $\varphi_T$ then $p_k$ is true at some \emph{finite} length of $w$. 
 We can easily construct a
Markov chain $M$ such that the set of finite traces of $M$ is $\Sigma^*$ ($\Sigma^*$ is the set of sets of proposition used).  We know that probability
measure of any finite trace of $M$ is greater than $0$. Thus, we can decide whether $\varphi_T$ is satisfiable iff we can decide 
$\Pr(M\models{\a v}(\varphi_T))>0$ for some valuation ${\a v}$. Hence, 
we  conclude that the emptiness problem of $V_{>0}(\varphi)$ is undecidable. 
\hfill $\quad\blacksquare$\end{proof}

\section{}
\textbf{Theorem 1.} For $\varphi \in$ pLTL(F,X), $V_{{>} 0}(\varphi) \neq \emptyset \mbox{ iff } \bar{{\a v}} \in V_{{>} 0}(\varphi)$ with $\bar{{\a v}}(x){=}m{\cdot}|\varphi|$.
\begin{proof} 
The direction from right to left is trivial.
Consider the other direction.
Let $\varphi$ be a pLTL(F,X)-formula and assume $V_{> 0}(\varphi) \neq \emptyset$.
By monotonicity, it suffices to prove that ${\a v}\in V_{>0}(\varphi)$ with ${\a v} \not\leq \bar{{\a v}}$ implies $\bar{{\a v}} \in V_{>0}(\varphi)$.
The proof proceeds in a number of steps. 
(1) We show that it suffices to consider formulas without disjunction. 
(2) We show that if path fragment
$\pi[0..l] \models \bar\varphi$, (where LTL(F,X)-formula $\bar\varphi$ is obtained from $\varphi$ by omitting all parameters from $\varphi$) 
then $\pi[0..l] \models {\a v}_l(\varphi)$  with ${\a v}_l(x) = l$ for every $x$.
(3) We construct a deterministic B\"uchi automaton (DBA) $A_{\bar\varphi}$ for $\bar\varphi$ such that its initial and final state are at most $| \bar\varphi |$ 
transitions apart.
(4) We show that reachability of a final state in the product of MC $M$ and DBA $A_{\bar\varphi}$ implies the existence of a finite path in $M$  
of length at most $m{\cdot}|\varphi|$ satisfying $\bar\varphi$.
\begin{enumerate}
\item
As disjunction distributes over $\wedge, {\pmb\bigcirc}, \de$, and $\de_x$, each formula can be written in disjunctive normal form.
Let $\varphi \equiv \varphi_1 \vee \ldots \vee \varphi_k$, where each $\varphi_i$ is disjunction-free. 
Evidently, $|\varphi_i | \leq | \varphi |$.
Assume ${\a v} \in V_{> 0}(\varphi)$.
Then, ${\a v} \in V_{> 0}(\varphi_i)$ for some $0 < i \leq k$.
Assuming the theorem holds for $\varphi_i$ (this will be proven below), $\bar { {\a v}}_i \in V_{> 0}(\varphi_i)$ with 
$\bar { {\a v}}_i(x) =  |\varphi_i|{\cdot} m$.
Since $\bar{ {\a v}} \geq \bar{ {\a v}}_i$, it follows by monotonicity that $\bar{ {\a v}} \in V_{> 0}(\varphi_i)$, 
and hence, $\bar{ {\a v}} \in V_{> 0}(\varphi)$.
It thus suffices in the remainder of the proof to consider disjunction-free formulas.
\item
For pLTL(F,X)-formula $\varphi$, let $\bar \varphi$ be the LTL(F,X)-formula obtained from $\varphi$ by replacing all occurrences of $\de_x$  by $\de$, 
e.g., for $\varphi = \de_x(a \wedge \de_y b)$, $\bar\varphi = \de(a \wedge\de b)$. 
We claim that $\pi[0...l] \models \bar\varphi$ implies $\pi[0...l] \models { {\a v}}_l(\varphi)$ with ${ {\a v}}_l(x) = l$ for all $x$. 
This is proven by induction on the structure of $\varphi$. 
The base cases $a$ and $\neg a$ are obvious.
For the induction step, conjunctions, $\X \varphi$ and $\de \varphi$ are straightforward.
It remains to consider $\de_x \, \varphi$.
Assume $\pi[0...l] \models \de \, \bar\varphi$. 
Thus, for some $i \leq l$, $\pi[i...l] \models \bar\varphi$. 
By induction hypothesis, $\pi[i...] \models{ {\a v}}_{il}(\varphi)$ with ${ {\a v}}_{il}(y)= l{-}i$ for each variable $y$ in $\varphi$. 
Thus, $\pi[0..l] \models {{ {\a v}}_l}(\de_x \, \varphi)$ with ${ {\a v}}_l(x)= l$ and for all $y$ in $\varphi$, ${ {\a v}}_l(y) = l$.
\item
We provide a DBA
$A_{\bar \varphi} = \la Q, \Sigma, \delta, q_0, F \ra$ with $\Sigma = 2^{\smAP}$ for each LTL(F,X)-formula $\bar \varphi$ using the construction from~\cite{DBLP:journals/tocl/AlurT04}.
We first treat $\bar \varphi = a$ and $\bar \varphi = \de a$.
As every LTL(F,X)-formula can be obtained from $\de (a \wedge \varphi)$, $\varphi_1 \wedge \varphi_2$ and $\X \varphi$, we then treat these inductive cases.
(Negations are treated similarly.)
For $\bar\varphi = a$, $A_{a} = \la \set{q_0,q_1}, \Sigma, \delta, q_0, \set{q_1} \ra$ with $\delta(q_0, a) = q_1$ and $\delta(q_1,\true) = q_1$.   
For $\bar\varphi = \de a$ , the DBA $A_{\de a} = \la \set{q_0,q_1}, \Sigma, \delta, q_0, \set{q_1} \ra$, where $\delta(q_0,a) = q_1$, $\delta(q_0,\neg a) = q_0$ and $\delta(q_1,\true) = q_1$. 
This completes the base cases.
For the three inductive cases, the DBA is constructed as follows.
\begin{enumerate}
\item
Let $A_{\bar \varphi} = \la Q, \Sigma, \delta, q_0, F \ra$.
$A_{\de (a \wedge \bar \varphi)} = \la Q \cup \set{q_0'}, \Sigma,\delta', q_0', F\ra$ where $q_0'$ is fresh, 
$\delta'(q,{\cdot}) = \delta(q,{\cdot})$ if $q \in Q$, $\delta'(q_0',a) = \delta(q_0,a)$, and $\delta'(q_0', \neg a) = q_0'$. 
\item For $\bar \varphi_1 \wedge \bar \varphi_2$, the DBA is a standard synchronous product of the DBA for $\bar \varphi_1$ and $\bar \varphi_2$.
\item Let $A_{\bar \varphi} = \la Q, \Sigma, \delta, q_0, F \ra$.
$A_{\X\! \bar \varphi} = \la Q \cup \set{q_0'}, \Sigma, \delta', q_0', F \ra$ where $q_0'$ is fresh,  $\delta'(q_0',a) = q_0$ for all 
$a\in \Sigma$ and $\delta'(q,a) = \delta(q,a)$ for every $q \in Q$.
\end{enumerate}
  
A few remarks are in order.
The resulting DBA have a single final state.
In addition, the DBA enjoy the property that the reflexive and transitive closure of the transition relation is a partial order~\cite{DBLP:journals/tocl/AlurT04}.
Formally, $q \preceq q'$ iff $q' \in\delta^*(q,w)$ for some $w \in \Sigma^\omega$. 
The diameter of $A_{\bar \varphi}$ is the length of a longest simple path from the initial to the final state.
This implies that
the diameter of $A_{\de (a \wedge \bar \varphi)}$ and  $A_{\X\! \bar \varphi}$ is $n{+}1$ where $n$ is this diameter of $A_{\bar \varphi}$, and
the diameter of $A_{\bar \varphi_1 \wedge \bar \varphi_2}$ is $n_1+n_2$ where $n_i$ is the diameter of $A_{\bar \varphi_i}$, $i \in \set{1,2}$.
%
\item
Let $\varphi \equiv \varphi_1 \vee \ldots \vee \varphi_k$, where each $\varphi_i$ is disjunction-free, with DBA $A_{\bar \varphi_i}$.
Evidently, $V_{>0}(\varphi) \neq \emptyset$ iff $V_{>0}(\varphi_i) \neq \emptyset$ for some disjunct $\varphi_i$.
Consider the product of MC $M$ and DBA $A_{\bar\varphi_i}$, denoted $M \otimes A_{\bar\varphi_i}$; see, e.g., \cite[Def.\ 10.50]{DBLP:books/daglib/0020348}.
By construction, $M\otimes A_{\bar\varphi_i}$ is partially ordered and has diameter at most $m{\cdot}|\varphi_i|$.
We have that $\Pr(M\models \bar \varphi_i) > 0$ iff an accepting state in $M \otimes A_{\bar\varphi_i}$ is reachable.
Thus, there exists a finite path $\pi[0..m{\cdot}|\varphi_i|]$ in $M$ with $\pi[0..m{\cdot}|\varphi_i] \models \bar \varphi$, or, $\pi[0..m{\cdot}|\varphi|] \models \bar\ve(\varphi)$.
This concludes the proof. 
\end{enumerate}
$M\otimes A_{\bar\varphi_i} $ can also be used to show that, if we have a valuation $\ve$ such that $\ve(x)>m{\cdot}|\varphi|$ and for all other variables $y\neq x$, $\ve(x)\leq m{\cdot}|\varphi|$ and 
$\ve \in V_{>0}(\varphi)$ then $\ve' \in V_{>0}(\varphi)$, where $\ve'(x)=m{\cdot}|\varphi|$ and for $y\neq x$, $\ve'(y)=\ve(y)$. The argument proceed as induction on $\bar\varphi_i$.
\hfill $\blacksquare$
\end{proof}

\section{}
\textbf{Proposition 5.} The problem $V_{>0}(\varphi)\neq \emptyset$ is NP-complete for $\varphi \in$ pLTL(F,X)
\begin{proof}
 Similar to the NP-hardness proof of satisfiability of LTL(F,X) formulas~\cite[Th.\ 3.7]{SistlaClarke85}, we give a polynomial reduction from the 3-SAT problem.
For 3-CNF formula $\phi$ with boolean variables $\{t_1,\hdots,t_n\}$, 
we define MC $M$ and pLTL(F,X) formula $\varphi$  such that $\phi$ is satisfiable iff  $V_{>0}(\varphi)$ is not empty.
Let 3-CNF formula $\phi = C_1 \wedge \hdots \wedge C_k$ with $C_i = d_{i1} \vee d_{i2} \vee d_{i3}$, 
where literal $d_{il}$ is either $t_k$ or $\neg t_k$.
Let MC $M = (S, \bfP, s_0, L)$ with $\AP = \set{C_i \mid 0 < i \leq k}$ be:
\begin{itemize}
\item 
$S = \set{s_i \mid 0 \leq i \leq n} \cup \set{t_i \mid 0 < i \leq n} \cup \set{\neg t_i \mid 0 < i \leq n}$
\item 
$\bfP(s_i,t_{i+1})>0$, $\mathbf P(s_i,\neg t_{i+1}) > 0$ for $0 \leq i < n$, $\bfP(t_i,s_i) > 0$ and $\bfP(\neg t_i,s_i) > 0$ for $0 < i \leq n$, 
and $\bfP(s_n,s_n)=1$ (the actual probabilities are not relevant), 
\item
$C_i \in L(t_j)$ iff $d_{il} = t_j$ for some $0 < l \leq 3$, and $C_i \in L(\neg t_j)$ iff $d_{il} = \neg t_j$ for some $0 < l \leq 3$, and
$L(s_j) = \emptyset$ for all $0 \leq j \leq n$.
\end{itemize}
Let pLTL(F,X)-formula $\varphi = \de_{y_1} \, C_1 \wedge \ldots \wedge \de_{y_k} \, C_k$. 
Then $\phi$ is satisfiable iff $ V_{>0}(\varphi)$ is not empty. Evidently, $M$ and $\varphi$ are obtained in polynomial time.

It remains to show membership in NP.
By the proof of Theorem~\ref{FX}, $V_{>0}(\varphi) \neq \emptyset$ iff there is a finite path of $M$ of length $m{\cdot}|\varphi|$ satisfying 
$\bar\varphi$. 
Thus, we non-deterministically select a path of $M$ of length $m{\cdot}|\varphi|$ and check (using standard algorithms) in polynomial time whether 
it satisfies $\bar\varphi$.  \hfill $\blacksquare$
\end{proof}

\section{}
\textbf{Proposition 6.} $|\min V_{>0}(\varphi) \, | \leq (N{\cdot}d)^{d{-}1}$.
\begin{proof}   
Let  $H= \set{0, \ldots, N}^d$. 
$(H, \leq)$ is a partially ordered set where $\leq$ is element-wise comparison. 
A subset $S^{(k)}$ of $H$ has rank $k$ if the summation of the coordinates of every element of $S$ is $k$.
By~\cite{K1978}, the largest set of incomparable elements (anti-chain) is given by $Z^{(k)}$ where $k$ is $N{\cdot}d/2$ if even, else $k$ is $(N{\cdot}d{-}1)/2$. 
Then $|Z| = {\lfloor N{\cdot}d/2\rfloor + d-1 \choose d-1}.$ 
\hfill $\blacksquare$\end{proof}

\section{}
\textbf{Theorem 3.}
$\min V_{>0}(\Box\de_x \, a) = n_0$ where $n_0 = \displaystyle\max \tiny{\left( n_{a,B}, \min_{\pi = s_0 \ldots s_n, s_n \models a_B} c(\pi) \right)}$ 
if $n_{a,B} < d_{a,B}(s_0)$ and $n_0 = n_{a,B}$ otherwise.

\begin{proof}
We show for $n \geq n_0$, $\Pr(\Box\de_n \, a) > 0$, and for $n < n_0$,
$\Pr(\Box\de_n \, a) = 0$.
Distinguish:
\begin{enumerate}
\item $n_{a,B} \geq d_{a,B}(s_0)$.  
Then, from $s_0$ an $a$-state in $B$ can be reached within $n_{a,B}$ steps, i.e., $\Pr(s_0 \models \de_{n_{a,B}} \, a_B) > 0$.
For this $a_B$-state, $s$, say, by Proposition~\ref{buchi} it follows $\Pr(s \models \Box \de_{n_{a,B}} \, a) = 1$.
Together this yields $\Pr(s_0 \models \Box \de_n \, a) > 0$ for each $n \geq n_{a,B} = n_0$.
For $n < n_0 = n_{a,B}$, it follows by Proposition~\ref{buchi} that $\Pr(s \models \Box \de_n \, a) = 0$ for every $a_B$-state $s$.
Thus, $\Pr(s_0 \models \Box\de_n \, a) = 0$.
\item
$n_{a,B} < d_{a,B}(s_0)$.  
As $B$ is accepting, $d_{a,B}(s_0) \neq \infty$.
Consider a simple path $\pi$ from $s_0$ to an $a$-state in $B$.
Let $c(\pi)$ be the maximal distance between two consecutive $a$-states along this path.
Then it follows $\Pr(s_0 \models \Box \de_k \, a) > 0$ where $k = \max (c(\pi), n_{a,B})$.
By taking the minimum $c_{min}$ over all simple paths between $s_0$ and $B$, it follows $\Pr(s_0 \models \Box \de_n \, a) > 0$ for each $n \geq n_0 = \max (n_{a,B}, c_{min})$ with $c_{min} = \min_{\pi \in \smPaths(s_0, \de a_B)} c(\pi)$.
For $n < n_0$, distinguish between $n_0 = n_{a,B}$ and $n_0 = c_{min}$.
In the former case, it follows (as in the first case) by Proposition~\ref{buchi} that $\Pr(s_0 \models \Box\de_n \, a) = 0$ for all $n \geq n_0$. 
Consider now $n_0 = c_{min} \geq n_{a,B}$.
Let $n < n_0$.
By contra-position.
Assume $\Pr(s_0 \models \Box \de_n \, a) > 0$.
Let $\pi = s_0 \ldots s_{1,a} \ldots s_{2,a} \ldots \ldots s_{k,a}$ be a finite path fragment in $M$ where $s_{i,a} \models a$ and $s_{k,a}$ is the first $a$-state along $\pi$ which belongs to $B$.
Then, by definition of the digraph $G_a$, the sequence $\pi = s_0 s_{1,a} s_{2,a} \ldots s_{k,a}$ is a path in $G_a$ satisfying $c(s_{i,a},s_{i{+}1,a}) \leq n$ for all $0 \leq k < n$. 
But then $c_{min} \leq n$.
Contradiction. \hfill $\blacksquare$
\end{enumerate} \end{proof}

\section{}
\textbf{Proposition 9.}
The sets $V_{>0}(\Box \de_{x} \, a)$ and $V_{=1}(\Box \de_{x} \, a)$ can be determined in polynomial time by a graph analysis of MC $M$.
\begin{proof} 
We argue that $\min V_{>0}(\Box\de_x \, a)$ can be determined in polynomial time.
The proof for $V_{=1}(\Box \de_x \, a)$ goes along similar lines and is omitted here.
We can determine both $n_{a,B}$ and $d_{a,B}(s_0)$ in linear time.
It remains to obtain $c_{min} = \min_{\pi = s_0 \ldots s_n, s_n \models a_B} c(\pi)$ in case $n_{a,B} < d_{a,B}(s_0)$.
This can be done as follows.
The distances $d(s,s')$, required for the function $c$ in the digraph $G_a = (V,E)$, can be obtained by applying Floyd-Warshall's all-pairs shortest path algorithm on the graph of $M$.
This takes ${\cal O}(m^3)$.
To obtain $c_{min}$, we use a cost function $F: V \to \Nats$ which is initially set to $0$ for initial state $s_0$ and $\infty$ otherwise.
Let $pQ$  be a min priority queue, initially containing all vertices of $G_a$, prioritized by the value of $F$. Algorithm~\ref{buchi-algo} finds 
$c_{min}$ in ${\cal O}(m^2{\cdot}\log m)$.
\begin{algorithm}
{\small
 \caption{Input: MC $M$  Output: $c_{min}$}\label{buchi-algo}
  \begin{algorithmic}[1]
    \State Initialize $F$, $\mbox{\sf found} := \false$ and $pQ$.
    \While{ $\left( \neg \mbox{\sf found} \mbox{ and } pQ \neq \emptyset \right)$}
      \State $u := \mbox{\sf pop}(pQ)$;  $\mbox{\sf found} := (a_B \in L(u))$;
      \For{$v \in pQ$} $F(v):= \min \left(F(v),\max(F(u),c(u,v)) \right)$ 
      \EndFor
    \EndWhile
  \end{algorithmic}
}
\end{algorithm}
Its correctness follows from the invariant $F(v) \leq \max(F(u),c(u,v))$. 
Using this we can find the minimum $n$ for which we can reach an accepting BSCC via a finite path satisfying $\Box\de_n \, a$. 
 \hfill $\blacksquare$
\end{proof}

\section{}
\noindent
\textbf{Theorem 5.} 
For $\varphi\in$ pLTL$_\de$, $V_{= 1}(\varphi) \neq \emptyset \mbox{ iff } \bar{ {\a v}} \in V_{= 1}(\bar\varphi)$ with $\bar{ {\a v}}(x) = m{\cdot}|\varphi|{\cdot}2^{|\varphi|}$.
\begin{proof} Consider the direction left to right.
 If there exists a reachable maximal SCC $C$ in the cross product  which is not complete then $\Pr(M\models \varphi)<1$. If every reachable maximal SCC is complete then the model checking task boils down to 
 reachability of such SCC. Thus the existence of a cycle before reaching a complete SCC implies that the probability measure of the set of paths satisfying $\varphi$ is strictly less than $1$ for any 
 value of the parameters. The largest cycle in the product can have at most $m{\cdot}|\varphi|{\cdot}2^{|\varphi|}$ states. Thus, if $\Pr(M\models \bar\ve(\varphi))$ is less than 1 then $V_{=1}(\varphi)$ is empty.
 \hfill $\blacksquare$\end{proof}
 
\end{document}

%% file: tlmacros.tex
\newcommand{\where}{\ | \ }

\renewcommand{\it}[1]{\singlearrow{#1}}  

\newcommand{\Paths}{\mbox{\sl Paths}}

\newcommand{\smAP}{\mbox{\scriptsize \sl AP}}
\newcommand{\smPaths}{\mbox{\scriptsize \sl Paths}}

\def\Acc{\mbox{\sl Acc}}

\renewcommand{\Pr}{\mbox{\rm Pr}}

\renewcommand{\Pr}{\mbox{\sl Pr}}

\def\bfP{\mathbf{P}}

\newcommand{\cl}{\mbox{\sl cl}}

\newcommand{\Not}{\mbox{$\, \neg \,$}}

\newcommand{\false}{\mbox{false}}
 
\newcommand{\true}{\mbox{true}} 
\newcommand{\AP}{\mbox{\sl AP}}

\newcommand{\Var}{\mbox{\sl Var}}

\newcommand{\set}[1]{\{ \, #1 \, \}}

\newcommand{\Nats}{\mbox{${\rm I\!N}$}}

\newcommand{\Next}{\bigcirc \, }
\newcommand{\Until}{\mbox{$\, {\sf U}\,$}}
\newcommand{\R}{\mbox{$\, {\sf R}\,$}}

\newcommand{\U}{\Until}

\newcommand{\Release}{\mbox{$\, {\sf R}\,$}}
\newcommand{\X}{\Next}

\newcommand{\trace}{\mbox{\sl trace}}

\renewcommand{\L}{{\cal L}}

\def\<{\langle}
\def\>{\rangle}